\documentclass[reqno]{amsart}

\usepackage{amsmath,amsthm,amssymb,amsfonts,enumerate,xcolor,bm}
\usepackage[utf8]{inputenc}

\numberwithin{equation}{section} 
\theoremstyle{plain}
\newtheorem{theo+}           {Theorem}      [section]
\newtheorem{prop+}  [theo+]  {Proposition}
\newtheorem{coro+}  [theo+]  {Corollary}
\newtheorem{lemm+}  [theo+]  {Lemma}
\newtheorem{defi+}  [theo+]  {Definition}
\newtheorem{conj+}  [theo+]  {Conjecture}

\theoremstyle{definition}
\newtheorem{rema+}  [theo+]  {Remark}
\newtheorem{prob+}  [theo+]  {Problem}
\newtheorem{exam+}  [theo+]  {Example}

\newenvironment{theorem}{\begin{theo+}}{\end{theo+}}
\newenvironment{proposition}{\begin{prop+}}{\end{prop+}}
\newenvironment{corollary}{\begin{coro+}}{\end{coro+}}
\newenvironment{lemma}{\begin{lemm+}}{\end{lemm+}}
\newenvironment{example}{\begin{exam+}}{\end{exam+}}

\newcommand{\ti}{\mathrm i}

\newcommand{\id}{\operatorname{id}}
\DeclareMathOperator{\Res}{Res}
 \newcommand{\Gampq}{\Gamma_{\!\!p,q}}
  \newcommand{\thp}{\theta_{p}}
\newcommand{\dup}{\textup{d}}

\makeatletter
\@namedef{subjclassname@2020}{\textup{2020} Mathematics Subject Classification}
\makeatother

\begin{document}

\baselineskip 18pt
\larger[2]
\title
[ $Q$-operators for the Ruijsenaars model] 
{ $Q$-operators for the Ruijsenaars model}
\author{Eric Rains}
\address{Department of Mathematics, California Institute of Technology, Pasadena CA 91125}
\email{rains@caltech.edu}
\urladdr{}
\author{Hjalmar Rosengren}
\address
{Department of Mathematical Sciences
\\ Chalmers University of Technology and University of Gothenburg\\SE-412~96 G\"oteborg, Sweden}
\email{hjalmar@chalmers.se}

\begin{abstract}
We prove that the Ruijsenaars model admits a one-parameter commuting family of $Q$-operators. The commutativity is equivalent to an elliptic hypergeometric integral transformation that was conjectured by Gadde et al., and  has an alternative interpretation in terms of $S$-duality for quiver gauge theories. We present two proofs of this conjecture, one using the elliptic Macdonald polynomials of Langmann et al., and one using known results on elliptic hypergeometric integrals. We also  explain 
how the Noumi--Sano operators appear as degenerations of  $Q$-operators.
\end{abstract}

\maketitle

\section{Introduction}

The Ruijsenaars model \cite{ru87} describes a system of relativistic quantum particles, whose Hamiltonian is given by a difference operator with elliptic coefficients. Ruijsenaars established the integrability of this model, by constructing higher order difference operators commuting with the Hamiltonian. 
To study the  joint eigenfunctions of these higher order Hamiltonians is a difficult problem, see e.g.\  \cite{lns,ru05,ru09a,ru09b}. 
The hyperbolic limit case is better understood, thanks to work of Halln\"as and Ruijsenaars \cite{hr14,hr15,hr21}
and Belousov et al.\ \cite{b1,b2,b3,b4,bk}.
The latter series of papers is to a large extent based on 
the existence of a commuting family of \emph{$Q$-operators} for the model.

Let us briefly discuss such operators on a formal level, working over $\mathbb R$ for simplicity.
Suppose that $H$ is a linear operator, which is symmetric on  $\mathrm L^2(X,d\mu)$, 
and $K=K(x,y)$ a function on $X\times X$ satisfying
the \emph{kernel function identity}
$H_xK=H_y K$.
If we define
$$(Qf)(y)=\int_X K(x,y) f(x)\,d\mu(x), $$ 
then  at least formally 
\begin{align}\nonumber(HQf)(y)&=\int H_yK(x,y)f(x)\,d\mu(x)= \int H_xK(x,y)f(x)\,d\mu(x)\\
\label{hqc}&=\int K(x,y)H_xf(x)\,d\mu(x)=(QHf)(y),\end{align}
that is, $[H,Q]=0$. Suppose now that $H$ has a complete system of orthonormal eigenfunctions
$(e_k)_{k\in K}$. Then, we can  formally construct kernel functions as
\begin{equation}\label{kex}K(x,y)=\sum_{k\in K} a_ke_k(x)e_k(y), \end{equation}
where $a_k$ are scalars. Since $Qe_k=a_ke_k$, any two operators from this family commute.

For the Ruijsenaars model, although a one-parameter family of explicit kernel functions $K=K_c$ is known \cite{ru05}, it is not obvious that the corresponding operators $Q=Q_c$ mutually commute. Until now, this has only be proved in the hyperbolic limit case \cite{b1}. 
The non-relativistic hyperbolic model (based on differential rather than difference operators) was treated in \cite{h}.

Before we describe our results, we provide some explanation of why the word  $Q$-operator is used in this context, see also \cite{ks,skl}.
Originally, $Q$-operators were introduced by Baxter in his solution of the eight-vertex model and XYZ spin chain \cite{bax1,bax2}. 
Just as  the operators $Q_c$ mentioned above, they form a one-parameter family  commuting with the Hamiltonian.
However, the similarities go  deeper.  Pasquier and Gaudin \cite{pg} found analogous operators for the periodic quantum Toda lattice, and discovered that a quasi-classical limit yields B\"acklund transformations for the corresponding classical system. Halln\"as and Ruijsenaars \cite{hr1} obtained similar results for the Ruijsenaars model, showing that  kernel functions can be used to obtain B\"acklund transformations of classical Calogero--Moser--Sutherland systems. 
Another common feature of $Q$-operators, which was crucial in Baxter's work, is that their eigenvalues satisfy a first-order equation in the parameter.
For the general Ruijsenaars model, nothing seems to be known about the $Q$-operator eigenvalues. However, in the non-relativistic hyperbolic  case, Halln\"as found that they satisfy a first-order difference equation  \cite{h}.

In the present work, we show that the $Q$-operators for the Ruijsenaars model commute in the sense that
 the integral kernel  of $[Q_c,Q_d]$ vanishes identically. We will not discuss commutativity in a 
Hilbert space setting, which would also involve analyzing the domain and range of the operators. 
 It turns out that vanishing of the commutator kernel  is equivalent to a transformation formula for elliptic hypergeometric integrals conjectured by 
 Gadde, Rastelli, Razamat and Yan \cite{grry}. In their work, this identity appears from a seemingly very different context, namely, it expresses $S$-duality between superconformal indices of  quiver gauge theories.
 The one-dimensional case of this integral identity is due to van de Bult \cite{vdb} and a hyperbolic version is proved in \cite{b1}. 
 Moreover, a closely related transformation between finite sums was obtained by Langer et al.\ \cite{lsw}. However, the general case has been open until now. 

We will in fact give two proofs of the GRRY conjecture. The first proof, given in \S \ref{fps}, relies on results of Langmann, Noumi and Shiraishi \cite{lns}. These authors show that the Ruijsenaars operators have joint eigenfunctions  in a space of formal power series in the elliptic nome $p$. The constant term $p=0$ gives the  Macdonald polynomials  \cite{m}. We prove that these elliptic Macdonald polynomials  are also eigenfunctions of the $Q$-operators. Hence, the $Q$-operators commute on a space of formal power series, which is enough to deduce that the commutator kernel vanishes. As a byproduct, we obtain a version of the identity \eqref{kex} for elliptic Macdonald polynomials, see Corollary \ref{ecc}.

In \S \ref{sps} we give a  proof based on known results for elliptic hypergeometric integrals. More precisely,  an integral transformation due to the first author \cite{rains10} implies that the commutator kernel vanishes when integrated against certain explicit test functions. We deduce from this that the kernel itself vanishes.

In \cite{ns}, Noumi and Sano introduced an alternative family of difference operators commuting with the Ruijsenaars Hamiltonian.  
In  \S \ref{nss} we show that the Noumi--Sano operators appear  as degenerate cases of the $Q$-operators, for singular values of the parameter $c$.
To be precise, this relation also involves a ``change of gauge'', which is a well-known symmetry of the Ruijsenaars model. A more complete picture appears if one considers four families of operators $Q_c$, $\tilde Q_c$, $H^{(k)}$ and $\tilde H^{(k)}$, where $\tilde Q_c$ are gauge transformed versions of the $Q$-operators and $H^{(k)}$ and $\tilde H^{(k)}$ two corresponding versions of the Noumi--Sano operators. As we describe in \S \ref{cps}, each pair of such operators commute thanks to a transformation formula for elliptic hypergeometric functions.

\section{Notation}

Throughout, the parameters $p$, $q$ and $t$ are assumed to satisfy $|p|<1$, $|q|<1$ and $|pq|<|t|<1$. Recall the standard notation
$$(x;q)_\infty=\prod_{j=0}^\infty(1-x q^j). $$
We will use the
multiplicative theta function
$$\theta_p(x)=(x;p)_\infty(p/x;p)_\infty $$
and Ruijsenaars's elliptic gamma function
\begin{equation}\label{gpg}\Gampq(x)=\prod_{j,k=0}^\infty\frac{1-p^{j+1}q^{k+1}/x}{1-p^jq^k x}.\end{equation}
We will often use the identities
$$\Gampq(qx)=\thp(x)\Gampq(x), $$
\begin{equation}\label{ginv} \Gampq(x)=\frac 1{\Gampq(pq/x)}.\end{equation}
Repeated arguments in one-variable functions  stand for products, e.g.,
$$(x,y;q)_\infty=(x;q)_\infty(y;q)_\infty,\qquad\Gampq(xy^\pm)=\Gampq(xy)\Gampq(x/y). $$

We will write
\begin{align}\notag\Lambda&=\Lambda^n=\{\lambda\in\mathbb N^n;\,\lambda_1\geq\dots\geq\lambda_n\geq 0\},\\
\label{pasa}\Lambda_\infty&=\Lambda_\infty^n=\{\lambda\in\mathbb Z^n;\,\lambda_1\geq\dots\geq\lambda_n\},\\
\notag\Lambda_0&=\Lambda_0^n=\{\lambda\in\mathbb N^n;\,\lambda_1\geq\dots\geq\lambda_n= 0\}.
\end{align}
On each of these sets, we define the dominance order $\lambda\leq \mu$ as
$$\lambda_1+\dots+\lambda_j\leq\mu_1+\dots+\mu_j,\qquad j=1,\dots, n,$$
 where we require that equality holds for $j=n$.

\section{Ruijsenaars operators and $Q$-operators}

The higher order Ruijsenaars operators are defined by
\begin{equation}\label{oru}D^{(k)}=\sum_{\substack{I\subseteq \{1,\dots,n\},\\ |I|=k}}\,\prod_{i\in I,j\notin I}\frac{\thp(tx_i/x_j)}{\thp(x_i/x_j)}\prod_{i\in I}T_{q,x_i}, \end{equation}
where $(T_{q,x}f)(x)=f(qx)$. 
It was proved by Ruijsenaars \cite{ru87} that
$$[D^{(k)},D^{(l)}]=0,\qquad 0\leq k,l\leq n.$$
The operators $D^{(k)}$ have  the kernel functions \cite{ru05}
\begin{equation}\label{rkf}K_c(\mathbf x;\mathbf y)=\prod_{1\leq i,j\leq n}\frac{\Gampq(cx_iy_j)}{\Gampq(ctx_iy_j)}, \end{equation}
that is,
\begin{equation}\label{rki}D^{(k)}_{\mathbf x}K_c=D^{(k)}_{\mathbf y}K_c. \end{equation}
Moreover, they are formally symmetric with respect to the pairing
\begin{equation}\label{rsp}\int_{\mathbb T^n}f(\mathbf x)g(\mathbf x^{-1})\prod_{1\leq i\neq j\leq n}\frac{\Gampq(tx_i/x_j)}{\Gampq(x_i/x_j)} \,|\dup \mathbf x|,\end{equation}
where
$$\mathbb T^n=\{\mathbf x\in\mathbb C^n;\,|x_1|=\dots=|x_n|=1\},$$
$$|\dup \mathbf x|=\prod_{j=1}^n\frac{\dup x_j}{2\pi\ti x_j}. $$
This suggests studying integral operators
\begin{equation}\label{tnq}f\mapsto \int_{\mathbb T^n}f(\mathbf x)
\prod_{1\leq i,j\leq n}\frac{\Gampq(cy_j/x_i)}{\Gampq(cty_j/x_i)}
\prod_{1\leq i\neq j\leq n}\frac{\Gampq(tx_i/x_j)}{\Gampq(x_i/x_j)} \,|\dup \mathbf x|. \end{equation}

We will focus on a refined version of the operators \eqref{tnq}, where the variables are constrained by the balancing condition
\begin{equation}\label{bxy}x_1\dotsm x_n=y_1\dotsm y_n.\end{equation} 
More precisely, let
$$\mathbb T^{n-1}_r=\{\mathbf x\in\mathbb C^n;\,|x_1|=\dots=|x_n|,\,x_1\dotsm x_{n}=r\}. $$
We then define
\begin{equation}\label{rqo}(Q_cf)(\mathbf y)=\int_{\mathbb T^{n-1}_{y_1\dotsm y_n}}f(\mathbf x)
\prod_{1\leq i,j\leq n}\frac{\Gampq(cy_j/x_i)}{\Gampq(cty_j/x_i)}
\prod_{1\leq i\neq j\leq n}\frac{\Gampq(tx_i/x_j)}{\Gampq(x_i/x_j)} \,|\dup \mathbf x|, \end{equation}
where now
$$|\dup \mathbf x|=\prod_{j=1}^{n-1}\frac{\dup x_j}{2\pi\ti x_j}.$$

The kernel function in \eqref{rqo} has poles when $x_i\in p^{\mathbb 
N}q^{\mathbb N}cy_j$ or $x_i\in p^{-\mathbb N-1}q^{-\mathbb N-1}cty_j$. 
We require that the first type of poles are inside the circle $|x_i|=|y_1\dotsm y_n|^{1/n}$ and the second type are outside.
This leads to the condition
$$|pq/t|<\left|\frac{cy_j}{(y_1\dotsm y_n)^{1/n}}\right|<1,\qquad j=1,\dots,n.$$
In particular, if we also assume $|y_1|=\dots=|y_n|$, then it is enough to require that
$$|pq/t|<|c|<1. $$
We can relax these conditions by deforming the domain of integration, see the proof of Proposition \ref{qhp}. However, 
 for simplicity we focus on the situation when we can integrate over a torus.

By  the same argument as in \eqref{hqc},  $[D^{(k)},Q_c]=0$. We provide some details in the Appendix. However, 
our main interest is in the identity $[Q_c,Q_d]=0$. To formulate our main result, we note that
$$(Q_cQ_df)(\mathbf y)=\int_{\mathbb T^{n-1}_{y_1\dotsm y_n}}f(\mathbf x)K_{cd}(\mathbf x;\mathbf y) \prod_{1\leq i\neq j\leq n}\frac{\Gampq(tx_i/x_j)}{\Gampq(x_i/x_j)}\,|\dup \mathbf x|,  $$
where
$$K_{cd}(\mathbf x;\mathbf y)=\int_{\mathbb T^{n-1}_{y_1\dotsm y_n}} \prod_{1\leq i\neq j\leq n}\frac{\Gampq(tz_i/z_j)}{\Gampq(z_i/z_j)} \prod_{i,j=1}^n\frac{\Gampq(cy_j/z_i,dz_i/x_j)}{\Gampq(cty_j/z_i,dtz_i/x_j)}\,|\dup\mathbf z|. $$

\begin{theorem}\label{mt}
Assume that  $|p|<1$, $|q|<1$, $|t|<1$, that $x_1\dotsm x_n=y_1\dotsm y_n=r^n$ and that 
\begin{equation}\label{xyc}|pq/t|<|cr/x_j|,\,|dr/x_j|,\,|cy_j/r|,\,|dy_j/r|<1,\qquad 1\leq j\leq n.\end{equation}
Then, the   identity $[Q_c,Q_d]=0$ holds in the sense that $K_{cd}(\mathbf x;\mathbf y)=K_{dc}(\mathbf x;\mathbf y)$.
\end{theorem}

As was mentioned in the introduction, an equivalent identity  was conjectured by Gadde et al.\ \cite[(Eq.\ D.2)]{grry}. To see this equivalence, one needs to apply \eqref{ginv} and use that, by \cite[(Eq.\ 2.14)]{grry}, their parameter $t$ (different from our $t$) satisfies $t^6=pq$.

Theorem \ref{mt} implies commutativity for the operators \eqref{tnq} and, more generally, for operators of the form
$$
(Q_{c,\phi}f)(\mathbf y)=\int_{\mathbb T^n}f(\mathbf x)\phi\left(\frac{x_1\dotsm x_n}{y_1\dotsm y_n}\right)
\prod_{1\leq i,j\leq n}\frac{\Gampq(cy_j/x_i)}{\Gampq(cty_j/x_i)}
\prod_{1\leq i\neq j\leq n}\frac{\Gampq(tx_i/x_j)}{\Gampq(x_i/x_j)} \,|\dup \mathbf x|,
$$
where $\phi$ is an integrable function. 

\begin{corollary}
Assume that $|p|<1$, $|q|<1$, $|t|<1$, $|pq/t|<|c|<1$ and $|pq/t|<|d|<1$. Then, the  integral kernel 
of the commutator $[Q_{c,\phi},Q_{d,\psi}]$ vanishes for $\mathbf x,\,\mathbf y\in\mathbb T^n$.
\end{corollary}

\begin{proof}
The integral kernel of $Q_{c,\phi}Q_{d,\psi}$ is
$$\int_{\mathbb T^{n}}\phi(Z/Y)\psi(X/Z)
 \prod_{1\leq i\neq j\leq n}\frac{\Gampq(tz_i/z_j)}{\Gampq(z_i/z_j)} \prod_{i,j=1}^n\frac{\Gampq(cy_j/z_i,dz_i/x_j)}{\Gampq(cty_j/z_i,dtz_i/x_j)}\,|\dup\mathbf z|,$$
 where we use notation such as $X=x_1\dotsm x_n$. We split this as an outer integral over $Z\in\mathbb T$ and and inner integral over $\mathbf z\in\mathbb  T^n_{Z}$.
 We also  introduce the variables $\xi_j=x_jZ^{1/n}/X^{1/n}$, $\eta_j=y_jZ^{1/n}/Y^{1/n}$, so that
 $\xi_1\dotsm\xi_n=\eta_1\dotsm\eta_n=Z$. The kernel then takes the form
 $$\int_{\mathbb T} \phi(Z/Y)\psi(X/Z)K_{cY^{1/n}/Z^{1/n},dZ^{1/n}/X^{1/n}}(\xi;\eta)|\dup Z|.$$
Making the change of variables $Z\mapsto XY/Z$ and applying Theorem \ref{mt}, it follows that the integral kernel is symmetric under interchanging $c$ and $d$.
\end{proof}

\section{First proof of commutativity}
\label{fps}

Our first proof of Theorem \ref{mt} is based on the elliptic Macdonald polynomials introduced in \cite{lns}. We first
 recall some standard facts about the Macdonald polynomials $P_\lambda(\mathbf x)=P_\lambda(\mathbf x;q,t)$, where $\mathbf x=(x_1,\dots,x_n)$. They are originally defined for partitions $\lambda$, but using the  property 
\begin{equation}\label{msp}P_{\lambda+(m)^n}(x_1,\dots,x_n)=(x_1\dotsm x_n)^m P_\lambda(x_1,\dots,x_n) \end{equation}
 one can extend the definition  to
$\lambda\in\Lambda_\infty$,  see \eqref{pasa}.
Under our standing assumption $|q|,\,|t|<1$, they satisfy the  orthogonality relations
\begin{equation}\label{mor}\int_{\mathbb T^n}P_\lambda(\mathbf x)P_\mu(\mathbf x^{-1})\prod_{1\leq i\neq j\leq n}\frac{(x_i/x_j;q)_\infty}{(tx_i/x_j;q)_\infty}\,|\dup \mathbf x|= N_\lambda\delta_{\lambda,\mu},\qquad \lambda,\,\mu\in\Lambda_\infty, \end{equation}
and the Cauchy identity 
\begin{equation}\label{ci}\prod_{i,j=1}^n\frac{(tx_iy_j;q)_\infty}{(x_iy_j;q)_\infty}=\sum_{\lambda\in\Lambda}b_\lambda P_\lambda(\mathbf x)P_\lambda(\mathbf y),\end{equation}
where $N_\lambda$ and $b_\lambda$ are explicitly known constants and the series \eqref{ci} converges for $\max_{i,j}(|x_iy_j|)\leq 1$.

We will need the following  consequence of 
\eqref{mor} and \eqref{ci}.

\begin{lemma}\label{pintl}
Let $\lambda\in\Lambda_\infty$ and assume that
 $|cy_j/(y_1\dotsm y_n)^{1/n}|<1$ for
   $1\leq j\leq n$. Then,
   \begin{equation}\label{pinti}\int_{\mathbb T^{n-1}_{y_1\dotsm y_n}}P_\lambda(\mathbf x)
\prod_{i,j=1}^n\frac{(cty_j/x_i;q)_\infty}{(cy_j/x_i;q)_\infty} \prod_{1\leq i\neq j\leq n}\frac{(x_i/x_j;q)_\infty}{(tx_i/x_j;q)_\infty}
 \,|\dup \mathbf x|=\phi_\lambda(c)P_\lambda(\mathbf y), \end{equation}
   where $\phi_\lambda$ is analytic in $|c|<1$. 
\end{lemma}

\begin{proof}
It is clear that the left-hand side of \eqref{pinti}
is analytic in $c$ under the stated assumption. 
It follows from \eqref{ci} that it has the Taylor expansion
$$\sum_{m=0}^\infty c^m\sum_{\mu\in\Lambda,\,|\mu|= m}b_\mu  P_\mu(\mathbf y)
\int_{\mathbb T^{n-1}_{y_1\dotsm y_n}}P_\lambda(\mathbf x)
P_\mu(\mathbf x^{-1}) \prod_{1\leq i\neq j\leq n}\frac{(x_i/x_j;q)_\infty}{(tx_i/x_j;q)_\infty}
 \,|\dup \mathbf x|.$$
 Making the change of
 variables $x_j\mapsto(y_1\dotsm y_n)^{1/n}x_j$,
this can be written
\begin{equation}\label{cie}\sum_{m=0}^\infty c^m\sum_{\mu\in\Lambda,\,|\mu|= m}b_\mu  P_\mu(\mathbf y)(y_1\dotsm y_n)^{(|\lambda|-|\mu|)/n} I_{\lambda,\mu},\end{equation}
where 
$$I_{\lambda,\mu}=\int_{\mathbb T^{n-1}}P_\lambda(\mathbf x)
P_\mu(\mathbf x^{-1}) \prod_{1\leq i\neq j\leq n}\frac{(x_i/x_j;q)_\infty}{(tx_i/x_j;q)_\infty}
 \,|\dup \mathbf x|.$$
Since the result must be independent of the choice of the root $(y_1\dotsm y_n)^{1/n}$,  $I_{\lambda,\mu}$ vanishes unless  $|\lambda|-|\mu|\in n\mathbb Z$. 
Assuming now that $|\mu|=|\lambda|+kn$, $k\in\mathbb Z$, consider \eqref{mor} with $\mu\mapsto\mu-(k)^n$. 
It can be written 
$$N_{\lambda}\delta_{\lambda,\mu-(k)^n}=
\int_{|s|=1}\int_{\mathbf x\in\mathbb T^{n-1}_s}P_\lambda(\mathbf x)P_{\mu-(k)^n}(\mathbf x^{-1})\prod_{1\leq i\neq j\leq n}\frac{(x_i/x_j;q)_\infty}{(tx_i/x_j;q)_\infty}\,|\dup \mathbf x|\,|\dup s|.$$
Making as before the change of variables
$x_j\mapsto s^{1/n}x_j$ and using \eqref{msp},
the integration in $s$ becomes trivial and 
the inner integral reduces to $I_{\lambda,\mu}$.
Hence, for fixed $m=|\lambda|+kn$, only the 
term $\mu=\lambda+(k)^n$ contributes to the 
inner sum in \eqref{cie}. The condition $\mu\in\Lambda$
 gives $k\geq -\lambda_n$.
We conclude that \eqref{pinti} holds with
$$\phi_\lambda(c)=N_\lambda\sum_{k=-\lambda_n}^\infty b_{\lambda+(k)^n}c^{|\lambda|+kn}. $$
This must converge under the stated assumption on $c$, which reduces to $|c|<1$ in the case when $|y_1|=\dots=|y_n|$.
\end{proof}

Although we will not need it, we mention that  the function $\phi_\lambda$ can be expressed in terms of the basic hypergeometric series ${}_n\phi_{n-1}$.

We  now recall some relevant results from
 \cite{lns}.
Consider the  space 
$$V=\mathbb C[x_1^\pm,\dots,x_n^\pm]^{\mathrm S_n}[[p]].$$
The Ruijsenaars operators $D^{(k)}$ act on $V$ and have joint eigenfunctions
$(\mathbf P_\lambda)_{\lambda\in\Lambda_\infty}$ of the form
\begin{equation}\label{gmx}\mathbf P_\lambda(\mathbf x;p)=\sum_{k=0}^\infty\sum_{\substack{\mu\in \Lambda_\infty,\\\mu\leq\lambda+k\phi}} C_{\lambda\mu}^{(k)}\,m_\mu(\mathbf x) p^k,\end{equation}
where
\begin{equation}\label{diag}C_{\lambda\lambda}^{(k)}=\delta_{k,0}. \end{equation}
Here, $m_\mu$ are the monomial symmetric Laurent polynomials, that is, $m_\mu(\mathbf x)$ is the sum of all distinct permutations of $x_1^{\mu_1}\dotsm x_n^{\mu_n}$, and  $\phi=(1,0,\dots,0,-1)$. It is easy to check that  $\mu\leq\lambda+k\phi$  implies $\lambda_n-k\leq\mu_j\leq\lambda_1+k$ for all $j$. 
In particular, the sum over $\mu$ in \eqref{gmx} has finite support, so that indeed $\mathbf P_\lambda\in V$. 
The constant term    $\mathbf P_\lambda(\mathbf x;0)$ equals the  Macdonald polynomial $P_\lambda(\mathbf x)$. We refer to $\mathbf P_\lambda$ as \emph{elliptic Macdonald polynomials}.
The authors of \cite{lns} prove that, under certain conditions, the series \eqref{gmx} is convergent. However, for our purposes it suffices to consider it as a formal series.


\begin{lemma}\label{sbl}
The elliptic Macdonald polynomials $(\mathbf P_\lambda)_{\lambda\in \Lambda_\infty}$
form a Schauder basis for $V$ as a module over $\mathbb C[[p]]$. That is, every element in $V$ can be written as
\begin{equation}\label{cpe}\sum_{\lambda\in \Lambda_\infty}A_\lambda(p) \mathbf P_\lambda(\mathbf x;p)  \end{equation} 
for unique $A_\lambda\in\mathbb C[[p]]$, where the sum is convergent as a formal power series. 
\end{lemma}

\begin{proof}
It is clear that $(m_\lambda)_{\lambda\in \Lambda_\infty}$ form a Schauder basis for $V$. By \eqref{diag}, we can write 
\eqref{gmx} as
$$m_\lambda(\mathbf x)=\mathbf P_\lambda(\mathbf x;p)-\sum_{k=0}^\infty\sum_{\mu\neq\lambda}
 C_{\lambda,\mu}^{(k)}\, m_\mu(\mathbf x) p^k.$$
We claim that iterating this relation
gives an expansion of $m_\lambda$ in $\mathbf  P_\mu$, which converges as a formal power series. 
If that was not the case, we could iterate it indefinitely and still keep the exponent of $p$ bounded. 
We would then have to pick $k=0$ in all except finitely many steps. In the remaining steps, we would apply it with $\mu<\lambda$ in the dominance order. As we saw above, the set of such $\mu\in\Lambda_\infty$ is finite, so the process must terminate. Hence, we end up with an expansion
$$m_\lambda(\mathbf x)=\sum_{k=0}^\infty\sum_{\mu\in\Lambda_\infty}
 D_{\lambda,\mu}^{(k)}\,\mathbf  P_\mu(\mathbf x;p) p^k,$$
where the inner sum has finite support. This proves the existence of the expansion \eqref{cpe}. 
For uniqueness, assume that \eqref{cpe} equals zero. After dividing by a power of $p$, we may assume that $A_\lambda(0)\neq 0$ for at least one $\lambda$. Then the uniqueness part follows from the linear independence of Macdonald polynomials.
\end{proof}

When $f\in V$, we define $Q_cf$ by applying \eqref{rqo} termwise to the formal power series expansion.

\begin{lemma}\label{qvl}  The operators $Q_c$ preserve the space $V$.\end{lemma}

\begin{proof}
We need to prove that if $f$ is a symmetric Laurent polynomial, then $Q_cf$ is a power series in $p$, whose coefficients are again symmetric Laurent polynomials.
 It is clear from \eqref{gpg} that 
$$\Gampq(x)=\frac 1{(x;q)_\infty}\sum_{k=0}^\infty\phi_k(x)p^k,$$
where $\phi_k\in\mathbb C[x^\pm]$. It follows that the integrand in \eqref{rqo} can be expanded as
$$\prod_{1\leq i\neq j\leq n}\frac{(x_i/x_j;q)_\infty}{(tx_i/x_j;q)_\infty}
\prod_{i,j=1}^n\frac{(cty_j/x_i;q)_\infty}{(cy_j/x_i;q)_\infty}\cdot\Phi(\mathbf x;\mathbf y;p),
 $$
where $\Phi$ is a formal power series in $p$ whose coefficients are Laurent polynomials, separately symmetric in $\mathbf x$ and in $\mathbf y$.
Hence, it is enough to show that if $f$ is a symmetric Laurent polynomial, then so is
$$ \int_{\mathbb T^{n-1}_{y_1\dotsm y_n}}f(\mathbf x)\prod_{1\leq i\neq j\leq n}\frac{(x_i/x_j;q)_\infty}{(tx_i/x_j;q)_\infty}
\prod_{i,j=1}^n\frac{(cty_j/x_i;q)_\infty}{(cy_j/x_i;q)_\infty} 
 \,|\dup \mathbf x|,$$
considered as a function of $\mathbf y$.
It suffices to take
 $f=P_\lambda$, for  $\lambda\in\Lambda_\infty$.
The result then follows from Lemma \ref{pintl}.
       \end{proof}

  It follows from the computation in the Appendix that $[D^{(k)},Q_c]f=0$ for $f\in\mathbb C[\mathbf x^{\pm}]$. Since both $D^{(k)}$ and $Q_c$ are defined to act termwise on formal power series, this relation extends to $f\in V$. 
We will use it to prove that the elliptic Macdonald polynomials are joint eigenfunctions of the $Q$-operators.
 
 \begin{corollary}\label{qec}
 We have
 $$Q_c\mathbf  P_\lambda(\mathbf x;p)=A_\lambda(p)\mathbf P_\lambda(\mathbf x;p),\qquad \lambda\in \Lambda,$$
 for some $A_\lambda\in\mathbb C[[p]]$. In particular, $[Q_c,Q_d]=0$ as an operator identity on $V$.
 \end{corollary}
 
 \begin{proof} It follows from Lemma \ref{sbl} and Lemma \ref{qvl} that we can write
$$Q_c \mathbf P_\lambda(\mathbf x;p)=\sum_{\mu\in \Lambda_\infty} A_{\lambda,\mu}(p)\mathbf P_\mu(\mathbf x;p),\qquad \lambda\in \Lambda_\infty, $$
with convergence as a formal power series. It is proved in \cite{lns} that
$D^{(k)}  \mathbf P_\lambda=\varepsilon_\lambda^{(k)}\mathbf  P_\lambda,$
 for some $\varepsilon_\lambda^{(k)}\in\mathbb C[[p]]$. 
 The constant terms $\varepsilon_\lambda^{(k)}(0)$ are the Macdonald eigenvalues
 $$ \varepsilon_\lambda^{(k)}(0)=e_k(t^{n-1}q^{\lambda_1},t^{n-2}q^{\lambda_2},\dots,q^{\lambda_n}),
 $$
 where $e_k$ denote elementary symmetric polynomials. 
 Identifying the coefficient of $\mathbf P_\mu$ in the identity $[D^{(k)},Q_c]\mathbf P_\lambda=0$  gives
 $$A_{\lambda,\mu}(p)\left(\varepsilon_\lambda^{(k)}(p)-\varepsilon_\mu^{(k)}(p)\right)=0.$$
 If $a_{\lambda,\mu}$ is the leading coefficient of $A_{\lambda,\mu}$, it follows that
 $$a_{\lambda,\mu}\left(\varepsilon_\lambda^{(k)}(0)-\varepsilon_\mu^{(k)}(0)\right)=0,\qquad 0\leq k\leq n.$$
 By \cite[Lemma 5.4]{noumi}, under our assumptions $|q|,\,|t|<1$, we can conclude that if $\lambda\neq \mu$ then $a_{\lambda,\mu}=0$.
 Hence, $A_{\lambda,\mu}$ is supported on $\lambda=\mu$. 
 \end{proof}

 We can now prove  that the integral kernel of $[Q_c,Q_d]$ vanishes.

 \begin{proof}[First proof of Theorem \ref{mt}]
 Let 
 $$K(\mathbf x;\mathbf y)=(K_{cd}(\mathbf x;\mathbf y)-K_{dc}(\mathbf x;\mathbf y))\prod_{1\leq i\leq j\leq n}
 \frac{\Gampq(tx_i/x_j)}{\Gampq(x_i/x_j)}.
 $$ 
By Corollary \ref{qec},
$$\int_{\mathbb T^{n-1}_{y_1\dotsm y_n}} K(\mathbf x;\mathbf y)f(\mathbf x)\,|\dup\mathbf x|=0$$
for any  symmetric Laurent polynomial $f$. 
It follows from elementary Fourier analysis that $K$ vanishes for $\mathbf x\in \mathbb T^{n-1}_{y_1\dotsm y_n}$. By analytic continuation,
it vanishes in the whole region \eqref{xyc}, assuming $x_1\dotsm x_n=y_1\dotsm y_n$.
 \end{proof}

 We can also give a version of \eqref{kex} for elliptic Macdonald polynomials.

 \begin{corollary}\label{ecc}
 Define $K^{(m)}$ by the Laurent expansion
\begin{equation}\label{kce}K_c(\mathbf x;\mathbf y)=\sum_{m=-\infty}^\infty
K^{(m)}(\mathbf x;\mathbf y)c^m\end{equation}
in the annulus
\begin{equation}\label{kcc}|pq/t|<|cx_iy_j|<1,\qquad 1\leq i,j\leq n.\end{equation}
Then, there exist $B_\lambda \in \mathbb C[[p]]$ such that
 \begin{equation}\label{kpe}K^{(m)}(\mathbf x;\mathbf y)=\sum_{\lambda\in\Lambda_\infty,\,|\lambda|=m}B_\lambda(p)\mathbf P_\lambda(\mathbf x;p)\mathbf P_\lambda(\mathbf y;p), \end{equation}
with convergence as a formal power series in $p$. 
 \end{corollary}

\begin{proof}
Let us temporarily write  $d=p/c$.  We can then expand
$$K_c=\prod_{i,j=1}^n\prod_{m=0}^\infty\frac{(c^{m+1}d^m  tx_iy_j,c^md^{m+1}q/x_iy_j;q)_\infty}{(c^{m+1}d^mx_iy_j,c^md^{m+1}q/tx_iy_j;q)_\infty}
=\sum_{k,l=0}^\infty\phi_{kl}(\mathbf x;\mathbf y)c^kd^l,$$
where $\phi_{k,l}$ are Laurent  polynomials and $|cx_iy_j|<1$, $|dq/tx_iy_j|<1$ for $1\leq i,j\leq n$. Equivalently, \eqref{kce} holds in the region \eqref{kcc}, with
$$K^{(m)}=\sum_{k=\max(0,-m)}^\infty\phi_{k,k+m}(\mathbf x;\mathbf y)p^k. $$
In particular, $K^{(m)}\in V\otimes V$.
It is clear that $K^{(m)}$ is homogeneous of degree
$m$ in $\mathbf x$ and in $\mathbf y$.
By
Lemma \ref{sbl}, it follows that
$$K^{(m)}(\mathbf x;\mathbf y)=
\sum_{\lambda,\mu\in\Lambda_\infty,\,|\lambda|=|\mu|=m}B_{\lambda,\mu}(p)\mathbf P_\lambda(\mathbf x;p)\mathbf P_\mu(\mathbf y;p),
 $$
 for some $B_{\lambda,\mu}\in\mathbb C[[p]]$. 

 Assume that $p$ is close enough to zero, so that
   \eqref{kcc} holds also if  $x_i$ and $y_j$ are multiplied by $q$. 
 We can then apply the expansion \eqref{kce} to both sides of the kernel function identity \eqref{rki}.
    Identifying the coefficient of $c^m$  gives
$$D^{(k)}_{\mathbf x}K^{(m)}=D^{(k)}_{\mathbf y}K^{(m)}.$$
Since this holds analytically near $p=0$, it also holds  in $V\otimes V$. Hence,
$$ B_{\lambda,\mu}(p)\left(\varepsilon_\lambda^{(k)}(p)-\varepsilon_\mu^{(k)}(p)\right)=0.$$
As in the proof of Corollary \ref{qec}, it follows that $B_{\lambda,\mu}$ is supported on $\lambda=\mu$.
\end{proof}

 It is tempting to write Corollary \ref{ecc} as
 $$K_c(\mathbf x;\mathbf y)=\sum_{\lambda\in\Lambda_\infty}c^{|\lambda|}B_\lambda(p)\mathbf P_\lambda(\mathbf x;p)\mathbf P_\lambda(\mathbf y;p).$$
 However,  we have not investigated the analytic convergence of this series.

\begin{example}
As an example, we verify \eqref{kpe} for $n=2$ and $m=0$, up to terms of order $p$.
It is straight-forward to compute
\begin{align}
\notag K^{(0)}&=1+\frac{(1-t)^2q}{t(1-q)^2}\left(2+\frac{x_1}{x_2}+\frac{x_2}{x_1}\right)\left(2+\frac{y_1}{y_2}+\frac{y_2}{y_1}\right)p+\mathcal O(p^2)\\
\label{kme}&=m_{0,0}\otimes m_{0,0}+\frac{(1-t)^2q}{t(1-q)^2}\left(2m_{0,0}+m_{1,-1}\right)\otimes\left(2m_{0,0}+m_{1,-1}\right)p+\mathcal O(p^2).
\end{align}
It follows from \eqref{gmx} and \eqref{diag} that
\begin{subequations}\label{pm}
\begin{align}
\mathbf P_{0,0}&=m_{0,0}+\alpha m_{1,-1}p+\mathcal O(p^2),\\
\mathbf P_{1,-1}&=m_{1,-1}+\beta m_{0,0}+\mathcal O(p),
\end{align}
\end{subequations}
for some scalars $\alpha$, $\beta$. 
Plugging this into the eigenvalue equation $D^{(1)}\mathbf P_\lambda=\varepsilon_\lambda(p) \mathbf P_\lambda$,
where
\begin{multline*}D^{(1)}=\frac{tx_1-x_2}{x_1-x_2}\left(1+\frac{(1-t)(x_1^2t-x_2^2)}{tx_1x_2}\,p\right)T_{q,x_1}\\
+\frac{tx_2-x_1}{x_2-x_1}\left(1+\frac{(1-t)(x_2^2t-x_1^2)}{tx_1x_2}\,p\right)T_{q,x_2}+\mathcal O(p^2), \end{multline*}
gives after a short computation
$$\alpha=\frac{(1-t)^2(1+t)q}{t(1-q)(1-t q)},\qquad \beta=\frac{(1-t)(1+q)}{1-tq}. $$
Inverting the relations \eqref{pm} gives
\begin{align*}
m_{0,0}&=\mathbf P_{0,0}+\alpha(\beta\mathbf P_{0,0}-\mathbf P_{1,-1})p+\mathcal O(p^2),\\
m_{1,-1}&=\mathbf P_{1,-1}-\beta\mathbf P_{0,0}+\mathcal O(p).
\end{align*}
Inserting these expressions into \eqref{kme} and simplifying we find that, in agreement with \eqref{kpe},
\begin{align*}K^{(0)}(\mathbf x;\mathbf y)&=\left(1+\frac{(t+1)(t-1)^2q(3tq+t-q-3)}{t(q-1)(tq-1)^2}\,p\right)\mathbf P_{0,0}(\mathbf x;p)\mathbf P_{0,0}(\mathbf y;p)\\
&\quad+\frac{(1-t)^2q}{t(1-q)^2}\,p\,\mathbf P_{1,-1}(\mathbf x;p)\mathbf P_{1,-1}(\mathbf y;p)
+\mathcal O(p^2).
 \end{align*}

\end{example}

\section{Second proof of commutativity}
\label{sps}

In our second proof, we work more directly with elliptic hypergeometric integrals. We will use the notation
\begin{equation}\label{inm}I_{n}^m(a_1,\dots,a_{m+n};b_1,\dots,b_{m+n})
=\kappa_n\int_{\mathbb T^{n-1}}\frac{\prod_{i=1}^n\prod_{j=1}^{n+m}\Gampq(a_jx_i,b_j/x_i)}{\prod_{1\leq i\neq j\leq n}\Gampq(x_i/x_j)}\,|\dup \mathbf x|,
 \end{equation}
 where $|a_j|<1$ and $|b_j|<1$ for $1\leq j\leq n+m$,
 $$a_1\dotsm a_{m+n}b_1\dotsm b_{m+n}=(pq)^m $$
 and 
 $$\kappa_n=\frac{(p;p)_\infty^{n-1}(q;q)_\infty^{n-1}}{n!}. $$
 We will also write
 \begin{equation}\label{jd}J_n(y_1,\dots,y_{2n};a,b)= \int_{\mathbb T^{n-1}}\frac{ \prod_{i=1}^n\prod_{j=1}^{2n}\Gampq(ax_iy_j,b/x_iy_j)}{\prod_{1\leq i\neq j\leq n}\Gampq(x_i/x_j,ab x_i/x_j)}\,|\dup\mathbf x|,\end{equation}
 where $|pq|<|ab|$, $|b|<|y_j|<|a^{-1}|$, $1\leq j\leq 2n$, and
 $$y_1\dotsm y_{2n}=1. $$
 The integral \eqref{inm}  is the $A$-type elliptic Dixon integral studied in \cite{rains10}.
 If we write  $t=pq/ab$, then \eqref{jd} contains the Selberg-type factor
 $$ \prod_{1\leq i\neq j\leq n}\frac{\Gampq(tx_i/x_j)}{\Gampq(x_i/x_j)}$$
 that we recognize from \eqref{rsp}. However, $J_n$ is different from the $A$-type elliptic Selberg integral studied in  \cite{s}.
 
 If $s^m=a_1\dotsm a_{n+m}=(pq)^m/b_1\dotsm b_{m+n}$
 and $|s|<|a_j|<1$, $|pq/s|<|b_j|<1$, $1\leq j\leq m+n$, then
one has the integral transformation \cite{rains10}
\begin{equation}\label{rait} I_{n}^m(\mathbf a;\mathbf b)
=\prod_{i,j=1}^{n+m}\Gampq(a_ib_j)\cdot I_{m}^n(pq/s \mathbf b;s/\mathbf a).\end{equation}
Although we will not need it, we mention that the
  parameter conditions can be relaxed at the expense of deforming the domain of integration.

It is easy to check that
$$K_{cd}(\mathbf x;\mathbf y)=J_n\left(\frac{1}{s x_1},\dots,\frac 1{s x_n},
\frac{s}{y_1},\dots,\frac{s}{y_n};s d, s c\right),
 $$
 where $s=\sqrt{pq/cdt}$. Hence,
the symmetry $K_{cd}=K_{dc}$ is equivalent to
\begin{equation}\label{jcs}J_n(y_1,\dots,y_{2n};a,b)=J_n(y_1,\dots,y_{2n};b,a). \end{equation}
Since the change of variables $x_j\mapsto x_j^{-1}$ 
gives the trivial symmetry
$$J_n(y_1,\dots,y_{2n};a,b)=J_n(y_1^{-1},\dots,y_{2n}^{-1};b,a), $$
 we can alternatively formulate 
\eqref{jcs} as
$$ J_n(y_1,\dots,y_{2n};a,b)=J_n(y_1^{-1},\dots,y_{2n}^{-1};a,b).$$
We stress that the relation $y_1\dotsm y_{2n}=1$ is essential for this to hold.

We will establish Theorem \ref{mt} by proving that
\begin{equation}\label{jft}\int_{\mathbb T^{2n-1}}J_n(\mathbf y;a,b)\phi(\mathbf y)\,|\dup\mathbf y|=\int_{\mathbb T^{2n-1}}J_n(\mathbf y;b,a)\phi(\mathbf y)\,|\dup\mathbf y| \end{equation}
for a sufficiently large class of test functions $\phi$.

\begin{lemma}\label{ftl}
Assume that
$abc^2d^2=pq$, 
 $|c|<|z_j|<|c^{-1}|$, $|d|<|w_j|<|d^{-1}|$, $j=1,\dots,n$. Then, 
 \eqref{jft} holds for 
$$\phi(\mathbf y)=\frac{\prod_{i=1}^{2n}\prod_{j=1}^n\Gampq(cz_j^{\pm}y_i,dw_j^\pm/y_i)}{\prod_{1\leq i\neq j\leq 2n}\Gampq(y_i/y_j)}.$$
\end{lemma}

\begin{proof}
Let $L$ denote the left-hand side of \eqref{jft}. 
Changing the order of integration gives
$$L
=\frac 1{\kappa_{2n}}\int_{\mathbb T^{n-1}} \frac{I_{2n}^n(a\mathbf x,c\mathbf z^\pm;b/\mathbf x,d\mathbf w^\pm)}{\prod_{1\leq i\neq j\leq n}\Gampq(x_i/x_j,abx_i/x_j)}\,
\,|\dup\mathbf x|.
 $$
The parameters are such that we can apply \eqref{rait}, with $s= ac^2$.  This gives
\begin{align}\nonumber L&=\frac 1{\kappa_{2n}}\prod_{i,j=1}^n\Gampq(cdw_i^{\pm}z_j^{\pm})
 \int_{\mathbb T^{n-1}} 
\frac{\prod_{i,j=1}^n\Gampq(adx_iw_j^\pm,bcx_i^{-1}z_j^\pm)}{\prod_{1\leq i\neq j\leq n}\Gampq(x_i/x_j)}\\
\label{ldi} &\quad\times  I_{n}^{2n}(c^2/\mathbf x,ac\mathbf z^\pm;d^2\mathbf x,bd\mathbf w^\pm)
\,|\dup\mathbf x|.\end{align}
Again changing the order of integration gives 
\begin{align*}L&=\frac 1{\kappa_{2n}}\prod_{i,j=1}^n\Gampq(cdw_i^{\pm}z_j^{\pm})
 \int_{\mathbb T^{n-1}} 
\frac{\prod_{i,j=1}^n\Gampq(acy_iz_j^\pm,bdy_i^{-1}w_j^\pm)}{\prod_{1\leq i\neq j\leq n}\Gampq(y_i/y_j)}\\
 &\quad\times  I_{n}^{2n}(d^2/\mathbf y,ad\mathbf w^\pm;c^2\mathbf y,bc\mathbf z^\pm)
\,|\dup\mathbf y|.\end{align*}
Replacing $y_j\mapsto x_j^{-1}$ gives \eqref{ldi} with $b$ and $a$ interchanged. This proves \eqref{jft}.
\end{proof}

\begin{proof}[Second proof of Theorem \ref{mt}]
In Lemma \ref{ftl}, we  make the change of parameters 
$$(a,b,c,d,p,q,\mathbf z,\mathbf w)\mapsto(r a,r b,rc,rd,r^3 p,r^3 q,r c\mathbf z,r^{-1} d^{-1} \mathbf w), $$
which  is consistent with the relation $abc^2d^2=pq$. After this change of variables,
\begin{equation}\label{wz}|w_j|<1<|z_j|,\qquad 1\leq j\leq n. \end{equation}

The integrand in \eqref{jd} has a formal power series
expansion
$$\frac{ \prod_{i=1}^n\prod_{j=1}^{2n}\Gamma_{r^3p,r^3q}(rax_iy_j,rb/x_iy_j)}{\prod_{1\leq i\neq j\leq n}\Gamma_{r^3p,r^3q}(x_i/x_j,r^2ab x_i/x_j)}=\Delta(\mathbf x)\Delta(\mathbf x^{-1})\sum_{k=0}^\infty r^k\phi_k(\mathbf x;\mathbf y),$$
where $\Delta(\mathbf x)=\prod_{1\leq i<j\leq n}(x_i-x_j)$ and each $\phi_k$ is a symmetric Laurent polynomial in $\mathbf x$ and in $\mathbf y$. Hence, we can write
$$J_n(\mathbf y;a,b)=\sum_{k=0}^\infty r^k\psi_k(\mathbf y), $$
where each $\psi_k$ is a symmetric Laurent polynomial. If 
 \eqref{jcs} fails analytically, it fails as a formal power series identity. That is, we can write
 \begin{equation}\label{jde}J_n(\mathbf y;a,b)-J_n(\mathbf y;b,a)=\psi(\mathbf y)r^{N}+\mathcal O(r^{N+1}) \end{equation}
 for some $N$ and some non-zero symmetric Laurent polynomial $\psi$. 
 
The formal power series expansion of $\phi$ has the form 
$$\phi(\mathbf y)=
\frac{\Delta(\mathbf y)\Delta(\mathbf y^{-1})}{\prod_{i=1}^{2n}\prod_{j=1}^{n}(1-y_i/z_j)(1-w_j/y_i)}
+\mathcal O(r).
  $$
 It then follows from \eqref{jft} that 
  $$\int_{\mathbf T^{2n-1}}
   \frac{\psi(\mathbf y)\Delta(\mathbf y)\Delta(\mathbf y^{-1})}{\prod_{i=1}^{2n}\prod_{j=1}^{n}(1-y_i/z_j)(1-w_j/y_i)}\,|\dup\mathbf y| 
   =0.$$
   We will show that this implies $\psi=0$, which contradicts our assumption
   that \eqref{jde} contains a non-zero term.

By \eqref{wz} and \cite[Eq.\ (I.4.2)]{m}, 
$$\prod_{i=1}^{2n}\prod_{j=1}^{n}\frac1{(1-y_i/z_j)(1-w_j/y_i)}=\sum_{\lambda,\,\mu\in\Lambda^n}
m_\lambda(\mathbf y)m_\mu(\mathbf y^{-1})
h_\lambda(\mathbf z^{-1})h_\mu(\mathbf w),
$$
where $m_\lambda$ and $h_\lambda$ denote, respectively, monomial and complete homogeneous symmetric polynomials.
It follows that
$$\int_{\mathbf T^{2n-1}}
   \psi(\mathbf y)\Delta(\mathbf y)\Delta(\mathbf y^{-1})m_\lambda(\mathbf y)m_\mu(\mathbf y^{-1})\,|\dup\mathbf y| 
   =0,\qquad \lambda,\mu\in\Lambda^n.  $$
Since $y_1\dotsm y_{2n}=1$, we can write
$m_\mu(\mathbf y^{-1})=m_{\tilde \mu}(\mathbf y)$,
 where 
 $$\tilde \mu=(\underbrace{k,\dots,k}_n,k-\mu_n,\dots,k-\mu_1) $$
and $k\in\mathbb Z$ is arbitrary. We now observe that 
any $\nu\in\Lambda^{2n}_\infty$ can be written $\nu=\lambda+\tilde\mu$ for $\lambda,\,\mu\in\Lambda^n$. For instance, we may take
\begin{align*}
\lambda&=(\nu_1-\nu_n,\nu_2-\nu_n,\dots,0),\\
\mu&=(\nu_n-\nu_{2n},\nu_n-\nu_{2n-1},\dots,\nu_n-\nu_{n+1}),\\
\tilde\mu&=(\nu_n,\dots,\nu_n,\nu_{n+1},\dots,\nu_{2n}).
\end{align*}
Then,
$$m_\lambda(\mathbf y)m_\mu(\mathbf y^{-1})= m_{\nu}(\mathbf y)+\sum_{\nu'\in\Lambda_\infty^{2n},\,\nu'<\nu} C_{\nu'} m_{\nu'}(\mathbf y).$$
Since $(m_\nu)_{\nu\in\Lambda^{2n}_\infty}$ form a basis for
the  symmetric Laurent polynomials in $\mathbf y$, it follows that 
$$\int_{\mathbf T^{2n-1}}
   \psi(\mathbf y)\Delta(\mathbf y)\Delta(\mathbf y^{-1})\phi(\mathbf y)\,|\dup\mathbf y| 
   =0  $$
 for any such polynomial $\phi$. This implies that $\psi$ vanishes identically.
\end{proof}

\section{Relation to Noumi--Sano operators}
\label{nss}

Noumi and Sano \cite{ns} introduced a family of difference operators that we will denote 
$$\tilde H^{(k)}=\sum_{\substack{\mu_1,\dots,\mu_n\geq 0,\\ \mu_1+\dots+\mu_n=k}}
\prod_{1\leq i<j\leq n}\frac{q^{\mu_j}\thp(q^{\mu_i-\mu_j}x_i/x_j)}{\thp(x_i/x_j)}\prod_{i,j=1}^n\frac{(tx_i/x_j;q,p)_{\mu_i}}{(qx_i/x_j;q,p)_{\mu_i}} \prod_{i=1}^n T_{q,x_i}^{\mu_i}.
 $$
 Here, the elliptic shifted factorials are defined by
 $$(x;q,p)_n=\frac{\Gampq (xq^n)}{\Gampq(x)}=\begin{cases}\thp(x)\thp(xq)\dotsm\thp(xq^{n-1}), & n\geq 0,\\
\big(\thp(xq^{-1})\thp(xq^{-2})\dotsm\thp(xq^{n})\big)^{-1}, & n<0.\end{cases} $$
The operators $\tilde H^{(k)}$ are polynomials in the Ruijsenaars operators $ D^{(k)}$, and vice versa. 
In particular, all these operators mutually commute. 
We will show that, in a certain sense, the Noumi--Sano operators are discrete degenerations of $Q$-operators.

It is well-known and easy to see that the Ruijsenaars operators \eqref{oru} satisfy
\begin{equation}\label{dgs}D^{(k)}=W^{-1} D^{(k)}\big|_{t\mapsto pq/t}W, \end{equation}
where
$$W(\mathbf x)=\prod_{1\leq i\neq j\leq n}\Gampq(tx_i/x_j). $$
Hence, the operators
$$H^{(k)}=(-1)^kq^{k(k+1)/2}t^{-kn}W^{-1}\tilde H^{(k)}\big|_{t\mapsto pq/t}W,$$
where the prefactor is chosen for convenience, are again polynomials in the Ruijsenaars operators.
Using the elementary identities
\begin{align*}
\prod_{1\leq i<j\leq n}\frac{q^{\mu_j}\thp(q^{\mu_i-\mu_j}x_i/x_j)}{\thp(x_i/x_j)}&=(-1)^{k(n-1)}q^{\frac{k(n-k-1)}2}
\frac{\prod_{j=1}^nq^{\frac{n\mu_j^2}2}x_j^{n\mu_j-k}}{\prod_{1\leq i\neq j\leq n}(x_i/x_j;q,p)_{\mu_i-\mu_j}}
,\\
 \prod_{i,j=1}^n(pqx_i/tx_j;q,p)_{\mu_i}&=(-tq^{-\frac12})^{kn}
\prod_{j=1}^nq^{-\frac{n\mu_j^2}2}x_j^{k-n\mu_j}\prod_{i,j=1}^n(qx_i/tx_j;q,p)_{\mu_i}
\end{align*}
it follows that
\begin{align*}H^{(k)}&=\sum_{\substack{\mu_1,\dots,\mu_n\geq 0,\\ \mu_1+\dots+\mu_n=k}}\prod_{1\leq i\neq j\leq n}
 \frac{(tx_i/x_j;q,p)_{\mu_i-\mu_j}}{(x_i/x_j;q,p)_{\mu_i-\mu_j}}
\prod_{i,j=1}^n\frac{(qx_i/tx_j;q,p)_{\mu_i}}{(qx_i/x_j;q,p)_{\mu_i}} \prod_{i=1}^n T_{q,x_i}^{\mu_i}.
\end{align*}

As we will see, 
the summand in $H^{(k)}$ can be viewed as a residue of the integral \eqref{rqo}, defining $Q_cf$, at the point $x_j=cq^{\mu_j}y_j$, $1\leq j\leq n$. The balancing condition \eqref{bxy} then formally implies $c=q^{-k/n}$, which is outside the domain of definition for $Q_c$. 
The following result gives a more precise statement.   We  formulate it for parameter conditions that are far from necessary, but make the proof reasonably simple.

\begin{proposition}\label{qhp}
Let  $f$ be a symmetric holomorphic function on $(\mathbb C\setminus\{0\})^n$ and let $k\in \mathbb N$. Assume that $|p|,\,|q|<1$
and $|p|<|q^{k-1}|$.  Fix a number $r$ with $1<r<|p^{-1}q^{k-1}|^{1/n}$ and assume that
$|pq|r<|t|<|q^k|r^{1-n}$. Finally, assume that $y_1,\dots,y_n$ satisfy
$$|y_i/y_j|< r,\qquad i\neq j,$$
 and are otherwise generic.  Then, the analytic continuation of $Q_cf$ in $c$  has at most a single pole at $c=q^{-k/n}$. Moreover
 $$\Res_{c=q^{-k/n}}(Q_cf)(\mathbf y)=-\frac{(n-1)!\,q^{-k/n}t^{kn}}{(p;p)_\infty^n(q;q)_\infty^n\Gampq(t)^n}(H^{(k)}f)(q^{-k/n}\mathbf y).$$
\end{proposition}

\begin{proof}
Initially, we only assume that  $|p|<1$, $|q|<1$ and that  $y_1,\dots,y_n$ are generic in the sense that
\begin{equation}\label{ygen}\frac{y_1\dotsm y_n}{y_{j_1}\dotsm y_{j_n}}\neq c^n p^{\mathbb N} q^{\mathbb N},\qquad 1\leq j_1,\dots,j_n\leq n \end{equation}
(in particular, $c\neq q^{-k/n}$) and constrained to an annulus
\begin{equation}\label{yac}a<|y_j|<b,\qquad 1\leq j\leq n. \end{equation}
The other parameter constraints will arise during the proof.

We  write the integral kernel of $Q_c$ as
$$M(\mathbf x;\mathbf y)=\prod_{1\leq i,j\leq n}\frac{\Gampq(cy_j/x_i)}{\Gampq(cty_j/x_i)}
\prod_{1\leq i\neq j\leq n}\frac{\Gampq(tx_i/x_j)}{\Gampq(x_i/x_j)}. $$
We will construct the analytic continuation of $Q_cf$ as an iterated integral
\begin{equation}\label{qac}
\int_{x_1\in\mathcal C_1}\dotsm\int_{x_{n-1}\in\mathcal C_{n-1}} f(\mathbf x)M(\mathbf x;\mathbf y)
\,\prod_{i=1}^{n-1}\frac{\dup x_i}{2\pi\ti x_i},
\end{equation}
where $x_n$ is determined from the balancing condition \eqref{bxy}. 
We take $\mathcal C_i$ to be  positively oriented contours such that
\begin{align}\label{ip}cy_jp^{\mathbb N}q^{\mathbb N}&\subseteq \operatorname{Int}(\mathcal C_i),\qquad 1\leq j\leq n,\\
\label{ep}\frac{y_1\dotsm y_n}{y_{j_1}\dotsm y_{j_{n-i}}x_1\dotsm x_{i-1}c^{n-i}}p^{-\mathbb N}q^{-\mathbb N}&\subseteq \operatorname{Ext}(\mathcal C_i),\qquad  1\leq j_1,\dots,j_{n-i}\leq n. \end{align}
Moreover, we require that on the domain of integration 
\begin{equation}\label{cac}\frac{a^{n-1}}{b^{n-2}|c|^{n-1}}<|x_i|<b|c|,\qquad 1\leq i\leq n,\end{equation}
which makes sense if
\begin{subequations}\label{tc}
\begin{equation}\label{abc}a^{n-1}<b^{n-1}|c|^n. \end{equation}

We will show that $\mathcal C_i$ exist by induction on $i$. 
We first note that the points \eqref{ip} are all mutually distinct from the points \eqref{ep}.
When $i=1$ this follows from \eqref{ygen} and when $i>1$ from the condition \eqref{ep} for $\mathcal C_{i-1}$. 
To prove that $\mathcal C_i$ can be chosen so that \eqref{cac} holds, we need to verify that
the points \eqref{ip} are smaller than $b|c|$ in modulus and the points \eqref{ep} larger than $a^{n-1}/b^{n-2}|c|^{n-1}$. The first inequality is immediate from \eqref{yac}. For the second inequality we note that we can always cancel a factor $y_j$ in \eqref{ep}.  Hence, we can estimate
$$\left|\frac{y_1\dotsm y_n}{y_{m_1}\dotsm y_{m_{n-i}}x_1\dotsm x_{i-1}c^{n-i}}\right|>\frac{a^{n-1}}{b^{n-i-1}(b|c|)^{i-1}|c|^{n-i}}
=\frac{a^{n-1}}{b^{n-2}|c|^{n-1}}, $$ 
where in the case $i>1$ we used \eqref{cac} for $\mathcal C_{i-1}$. 
 For the final contour $\mathcal C_{n-1}$, we also need to ensure that \eqref{cac} holds with $i=n$. 
 However, if we let  $\mathcal C_n=(y_1\dotsm y_n/x_1\dotsm x_{n-2})\mathcal C_{n-1}^{-1}$ denote the  contour of integration for $x_n$, then  \eqref{ip} for $i=n-1$ is equivalent to \eqref{ep} for $i=n$ and vice versa. 
 Hence, the case $i=n$ of \eqref{cac} follows in the same way as $i=n-1$. 

We now consider the position of the singularities of the integrand. 
The factor $\Gampq(cy_j/x_i)$ has poles at the points \eqref{ip}, which are inside 
 $\mathcal C_i$
 for $1\leq i\leq n-1$.
As we saw in the preceding paragraph, they are also
 inside $\mathcal C_n$. 
The factor $\Gampq(cty_j/x_i)^{-1}$ has poles at $x_i=cty_jp^{-\mathbb N-1}q^{-\mathbb N-1}$. If 
\begin{equation}\label{tca}b|pq|<a|t|,\end{equation}
 it follows from
\eqref{yac} and \eqref{cac} that these are outside $\mathcal C_i$ for $1\leq i\leq n$. 
The factor $\Gampq(tx_i/x_j)$ is analytic for $|tx_i/x_j|<1$. In  \eqref{rqo}, the domain of integration is included in that region.
If we assume
\begin{equation}\label{tcb}b^{n-1}|c^nt|<a^{n-1}, \end{equation}
\end{subequations}
it follows from \eqref{cac} that this also true for \eqref{qac}. Finally, we note that
\begin{equation}\label{dtp}\frac 1{\prod_{1\leq i\neq j\leq n}\Gampq(x_i/x_j)}=\prod_{1\leq i<j\leq n}\thp(x_i/x_j)\theta_q(x_j/x_i) 
\end{equation}
is analytic on $(\mathbb C\setminus\{0\})^n$.
This shows that, under the parameter conditions \eqref{tc}, all singularities in \eqref{qac} are located in the same way relative to the domain of integration as in \eqref{rqo}. Hence, \eqref{qac} defines an analytic continuation of $Q_cf$. 

We now take $c$ in a small neighborhood of $q^{-k/n}$ and write $r=b/a$. Then, \eqref{abc} holds automatically, whereas \eqref{tca} and \eqref{tcb} reduce to $|pq|r<|t|<|q^k|r^{1-n}$, which implies $r^n<|p^{-1}q^{k-1}|$. 
That is, we recover the  conditions used in the formulation of the proposition. 

Next, we let $c\rightarrow q^{-k/n}$. Consider first the integration over $\mathcal C_1$.
In the limit, the points
$$S=\big\{cy_jq^\mu;\, 1\leq j\leq n,\, 0\leq \mu\leq k\big\} $$
are included both in \eqref{ip} and \eqref{ep} for $i=1$. Let $\mathcal D_1$ be a contour defined exactly as $\mathcal C_1$, except that 
$S\subseteq\operatorname{Ext}(\mathcal D_1)$. We can then write
$$ \int_{x\in\mathcal C_1}=\int_{x\in\mathcal D_1}+2\pi\ti\sum_{z\in S}\Res_{x=z}.$$
If $x_1\in\mathcal D_1$ and the parameters $y_j$ are generic, the remaining integrals are regular at $c=q^{-k/n}$. However, if $x_1\in S$, then some points in $S$ are included both in \eqref{ip} and \eqref{ep} for $i=2$.  Repeating this argument leads to the decomposition
\begin{equation}\label{qrs}Q_cf=\sum_{z_1,\dots,z_{n-1}\in S}\Res_{\mathbf x=\mathbf z} \frac{f(\mathbf x)M(\mathbf x;\mathbf y)}{
 \prod_{i=1}^{n-1} x_i}+R,\end{equation}
where $R$ is regular at $c=q^{-k/n}$ and $x_n$ is  determined from \eqref{bxy}. 

If we write $z_j=cy_{\sigma(j)}q^{\mu_j}$, $1\leq j\leq n-1$, then it follows from \eqref{dtp} that only terms with $\sigma$ injective contribute to the sum in \eqref{qrs}. We first  consider the case $\sigma=\id$.
Then, the potential singularity at $c=q^{-k/n}$ comes from the factor
$$\Gampq(cy_n/x_n)\Big|_{\mathbf x=\mathbf z}=\Gampq(cx_1\dotsm x_{n-1}/y_1\dotsm y_{n-1})\Big|_{x_j=cy_jq^{\mu_j}}
=\Gampq(c^nq^{k-\mu_n}), $$
where $\mu_n=k-(\mu_1+\dots+\mu_{n-1})$. This is singular at $c=q^{-k/n}$ if $\mu_n\geq 0$, and we observe that
$$\Res_{c=q^{-k/n}}\Gampq(c^nq^{k-\mu_n})
=-\frac{q^{-\frac kn}}{nx_n}\,\Res_{x_n=cy_nq^{\mu_n}}\Gampq(cy_n/x_n)\Big|_{c=q^{-k/n}}.
 $$
 Hence, the resulting contribution to  $\Res_{c=q^{-k/n}}(Q_cf)$ is
 $$-\frac{q^{-\frac kn}}{n}\sum_{\substack{\mu_1,\dots,\mu_n\geq 0,\\ \mu_1+\dots+\mu_n=k}} f(\mathbf y q^{\bm \mu-k/n})\Res_{\mathbf x=c\mathbf y q^{\bm \mu} }
  \frac{M(\mathbf x;\mathbf y)}{\prod_{i=1}^{n}x_i}\Bigg|_{c=q^{-k/n}}.
   $$
   Since this is symmetric in $\mathbf y$,
   the terms corresponding to $\sigma\neq\id$ each add an identical contribution and we conclude that
   $$\Res_{c=q^{-k/n}} Q_cf=-q^{-\frac kn}(n-1)!\sum_{\substack{\mu_1,\dots,\mu_n\geq 0,\\ \mu_1+\dots+\mu_n=k}}
    f(\mathbf y q^{\bm \mu-k/n})
   \Res_{\mathbf x=c\mathbf y q^{\bm \mu} }
 \frac{M(\mathbf x;\mathbf y)}{\prod_{i=1}^nx_i}\Bigg|_{c=q^{-k/n}}
.$$

To finish the computation, we write
$$\Res_{\mathbf x=c\mathbf y q^{\bm \mu} }
 \frac{M(\mathbf x;\mathbf y)}{\prod_{i=1}^nx_i}=
 \Res_{\mathbf x=c\mathbf y }\frac{M(\mathbf x q^{\bm \mu};\mathbf y)}{M(\mathbf x ;\mathbf y)}\cdot
 \frac{M(\mathbf x;\mathbf y)}{\prod_{i=1}^nx_i}.$$
Since
$$\frac{M(\mathbf x q^{\bm \mu};\mathbf y)}{M(\mathbf x ;\mathbf y)}
=t^{kn}\prod_{i,j=1}^n\frac{(qx_i/cty_j;q,p)_{\mu_i}}{(qx_i/cy_j;q,p)_{\mu_i}}\prod_{1\leq i\neq j\leq n}\frac{(tx_i/x_j;q,p)_{\mu_i-\mu_j}}{(x_i/x_j;q,p)_{\mu_i-\mu_j}},$$
 the first factor is regular at $\mathbf x=c\mathbf y$ and equals
 $$ t^{kn}\prod_{i,j=1}^n\frac{(qy_i/ty_j;q,p)_{\mu_i}}{(qy_i/y_j;q,p)_{\mu_i}}\prod_{1\leq i\neq j\leq n}\frac{(ty_i/y_j;q,p)_{\mu_i-\mu_j}}{(y_i/y_j;q,p)_{\mu_i-\mu_j}}.$$
Finally, we have
\begin{align*}
\Res_{\mathbf x=c\mathbf y }
 \frac{M(\mathbf x;\mathbf y)}{\prod_{i=1}^nx_i}&=
\Res_{\mathbf x=c\mathbf y }
\prod_{1\leq i\neq j\leq n}
\frac{\Gampq(cy_j/x_i,tx_i/x_j)}{\Gampq(cty_j/x_i,x_i/x_j)}
\prod_{i=1}^n\frac{\Gampq(cy_i/x_i)}{x_i\Gampq(cty_i/x_i)}
\\
&=\frac 1{\Gampq(t)^n}\prod_{i=1}^n\Res_{x_i=cy_i}\frac{\Gampq(cy_i/x_i)}{x_i}=\frac{1}{(p;p)_\infty^n(q;q)_\infty^n\Gampq(t)^n}.
\end{align*}
This completes the proof.
\end{proof}

\section{$Q$-operators and elliptic hypergeometric functions}
\label{cps}

In view of \eqref{dgs}, it is natural to consider the gauge transformed operators
$$\tilde Q_c=W^{-1}Q_c\big|_{t\mapsto pq/t}W $$
or, more explicitly,
$$
(\tilde Q_cf)(\mathbf y)=\int_{\mathbb T^{n-1}_{y_1\dotsm y_n}}f(\mathbf x)
\frac{\prod_{1\leq i,j\leq n}\Gampq(cy_j/x_i,tx_i/cy_j)}{\prod_{1\leq i\neq j\leq n}\Gampq(ty_i/y_j,x_i/x_j)}
 \,|\dup \mathbf x|.
$$
It  follows from Theorem \ref{mt} that $[\tilde Q_c,\tilde Q_d]=0$, but it is not a priori clear that $[Q_c,\tilde Q_d]=0$. It turns out that this
follows from the integral transformation \eqref{rait}. In the hyperbolic case, an analogous observation was made in \cite{b4}.

\begin{proposition}\label{qhqp}
Assume that $|p|<1$, $|q|<1$,  $x_1\dotsm x_n=y_1\dotsm y_n=r^n$ and, for $1\leq j\leq n$,
$$|pq/t|<|cr/x_j|,\,|cy_j/r|<1,\qquad |t|<|dr/x_j|,\,|dy_j/r|<1. $$
Then
the integral kernel of $[Q_c,\tilde Q_d]$, considered as a function of $\mathbf x$ and $\mathbf y$, vanishes identically.
\end{proposition}

\begin{proof}
We compute
$$(Q_c\tilde Q_d f)(\mathbf y)=\int_{\mathbb T^{n-1}_{r^n}}\frac{f(\mathbf x)}{\prod_{1\leq i\neq j\leq n}\Gampq(x_i/x_j)} \,K_{1}(\mathbf x;\mathbf y) \,|\dup \mathbf x|, $$
where
\begin{align*}K_{1}(\mathbf x;\mathbf y)&= \int_{\mathbb T^{n-1}_{r^n}}\frac{\prod_{i,j=1}^n\Gampq(dz_i/x_j,pqz_i/cty_j,tx_j/dz_i,cy_j/z_i)}{\prod_{1\leq i\neq j\leq n}\Gampq(z_i/z_j)} \,|\dup \mathbf z|\\
&=\frac 1{\kappa_n}\,I_n^n(dr/\mathbf x,pqr/ct\mathbf y;t\mathbf x/dr,c\mathbf y/r).
\end{align*}
The product $\tilde Q_dQ_c$ is given by a similar expression, with $K_1$ replaced by
$$ K_{2}(\mathbf x;\mathbf y)=\frac 1{\kappa_n}
\prod_{1\leq i\neq j\leq n}\frac{\Gampq(tx_i/x_j)}{\Gampq(ty_i/y_j)}\,
I_n^n(cr/\mathbf x,tr/d\mathbf y;pq\mathbf x/ctr,d\mathbf y/r).
$$
It follows from \eqref{rait} that, under the given parameter conditions,  $K_1=K_2$.
\end{proof}

\def\arraystretch{1.3}
\begin{table}[h]
\caption{Elliptic hypergeometric transformations and $Q$-operator identities}\label{sit}
\begin{tabular}{|c|c|}
\hline
{\bf Kajihara--Noumi--Rosengren} & {\bf Langer--Schlosser--Warnaar}\\ 
 $[H,\tilde H]=0$ & $[H,H]=[\tilde H,\tilde H]=0$\\
$H_{\mathbf x}\tilde K=H_{\mathbf y}\tilde K,\quad \tilde H_{\mathbf x} K=\tilde H_{\mathbf y} K$ & $H_{\mathbf x}K=H_{\mathbf y} K,\quad \tilde H_{\mathbf x} \tilde K=\tilde H_{\mathbf y} \tilde K$\\
$[H,\tilde Q]=[\tilde H,Q]=0$ &$[H,Q]=[\tilde H,\tilde Q]=0$ \\
\hline
  {\bf Rains} & {\bf Gadde--Rastelli--Razamat--Yan}\\
 $[Q,\tilde Q]=0$ & $[Q,Q]=[\tilde Q,\tilde Q]=0$\\
$Q_{\mathbf x}\tilde K=Q_{\mathbf y}\tilde K,\quad \tilde Q_{\mathbf x} K=\tilde Q_{\mathbf y} K$ & $Q_{\mathbf x} K=Q_{\mathbf y} K,\quad \tilde Q_{\mathbf x}\tilde K=\tilde Q_{\mathbf y}\tilde  K$\\
\hline
\end{tabular}
\end{table}

In conclusion, we have two families of commuting $Q$-operators $Q_c$ and $\tilde Q_c$, and two corresponding families of Noumi--Sano operators $H^{(k)}$ and $\tilde H^{(k)}$. 
Any commutation relation between two such operators can be reduced to a transformation formula for elliptic hypergeometric functions. The same can be said about the corresponding kernel function identities, both for the  kernel function \eqref{rkf} and for its gauge transform
$$\tilde K_c(\mathbf x;\mathbf y)=\frac {K_c(\mathbf x;\mathbf y)\big|_{t\mapsto pq/t}}{W(\mathbf x)W(\mathbf y)}=\frac{\prod_{i,j=1}^n\Gampq(cx_iy_j,t/cx_iy_j)}{\prod_{1\leq i\neq j\leq n}\Gampq(tx_i/x_j,ty_i/y_j)}.$$
We summarize these relations in Table \ref{sit}.
The southeast corner contains Theorem~\ref{mt} and some closely related results. Here, for instance, the
 entry $[Q,Q]=0$ is a short-hand for the fact that the integral kernel of
 $[Q_c,Q_d]$ vanishes for some range of parameters. It is easy to see that this is formally equivalent to 
  the kernel function identity $(Q_c)_{\mathbf x}K_d(\mathbf x;\mathbf y)=(Q_c)_{\mathbf y}K_d(\mathbf x;\mathbf y)$. 
  These results follow  from the 
GRRY conjecture that we have now proved.
   Likewise, the southwest corner contains Proposition~\ref{qhqp} and its equivalent forms, which follow from the integral transformation \eqref{rait}.
By Proposition \ref{qhp}, the operators $H^{(k)}$ and $\tilde H^{(k)}$ can be obtained as degenerations of $Q_c$ and $\tilde Q_c$. Therefore, it is not surprising that the entries in the upper half of the table follow from finite-sum degenerations of the GRRY and Rains identity. For instance,
in \cite{hlnr}
it was observed  that  $[\tilde H^{(k)},\tilde H^{(l)}]=0$ follows from a discrete version of the GRRY identity due to Langer et al.\ \cite{lsw} and also that
the kernel function identity $\tilde H_\mathbf x^{(k)}K_c(\mathbf x;\mathbf y)=\tilde H_\mathbf y^{(k)}K_c(\mathbf x;\mathbf y)$ follows from 
a discrete version of \eqref{rait} due to Kajihara and Noumi \cite{kn} and to the second author \cite{ro}.

 \section*{Appendix: Commutation between $Q$-operators and Ruijsenaars operators}

We will prove that $[Q_c,D^{(k)}]=0$. 
As was noted by Ruijsenaars,  \eqref{rki} is equivalent to the  theta function identity
\cite{kn}
\begin{multline*} \sum_{\substack{I\subseteq \{1,\dots,n\},\\ |I|=k}}\,\prod_{i\in I,j\notin I}\frac{\thp(tx_i/x_j)}{\thp(x_i/x_j)}\prod_{i\in I,\,1\leq j\leq n}
\frac{\thp(cx_iy_j)}{\thp(ctx_iy_j)}\\
=\sum_{\substack{I\subseteq \{1,\dots,n\},\\ |I|=k}}\,\prod_{i\in I,j\notin I}\frac{\thp(ty_i/y_j)}{\thp(y_i/y_j)}\prod_{i\in I,\,1\leq j\leq n}
\frac{\thp(cx_jy_i)}{\thp(ctx_jy_i)}.
\end{multline*}
Using this identity, with $\mathbf x\mapsto \mathbf x^{-1}$, gives
\begin{align*}
(D^{(k)}Q_cf)(\mathbf y)&= \sum_{\substack{I\subseteq \{1,\dots,n\},\\ |I|=k}}\int_{\mathbb T^{n-1}_{y_1\dotsm y_nq^k}}f(\mathbf x)
 \prod_{i\in I,j\notin I}\frac{\thp(ty_i/y_j)}{\thp(y_i/y_j)}\prod_{i\in I,\,1\leq j\leq n}
\frac{\thp(cy_i/x_j)}{\thp(cty_i/x_j)}\\
&\quad\times\prod_{1\leq i,j\leq n}\frac{\Gampq(cy_j/x_i)}{\Gampq(cty_j/x_i)}
\prod_{1\leq i\neq j\leq n}\frac{\Gampq(tx_i/x_j)}{\Gampq(x_i/x_j)} \,|\dup \mathbf x|\\
&= \sum_{\substack{I\subseteq \{1,\dots,n\},\\ |I|=k}}\int_{\mathbb T^{n-1}_{y_1\dotsm y_nq^k}}f(\mathbf x)
 \prod_{i\in I,j\notin I}\frac{\thp(tx_j/x_i)}{\thp(x_j/x_i)}\prod_{i\in I,\,1\leq j\leq n}
\frac{\thp(cy_j/x_i)}{\thp(cty_j/x_i)}\\
&\quad\times\prod_{1\leq i,j\leq n}\frac{\Gampq(cy_j/x_i)}{\Gampq(cty_j/x_i)}
\prod_{1\leq i\neq j\leq n}\frac{\Gampq(tx_i/x_j)}{\Gampq(x_i/x_j)} \,|\dup \mathbf x|.
\end{align*}
The factors $\thp(x_i/x_j)^{-1}$ have singularities on the domain of integration, but these are cancelled by zeroes of $\Gampq(x_i/x_j)^{-1}$.
We now make the change of variables $x_i\mapsto qx_i$ for all $i\in I$. Then,
$$ \prod_{i\in I,\,1\leq j\leq n}
\frac{\thp(cy_j/x_i)}{\thp(cty_jx_i)}\prod_{1\leq i,j\leq n}\frac{\Gampq(cy_j/x_i)}{\Gampq(cty_j/x_i)}\mapsto\prod_{1\leq i,j\leq n}\frac{\Gampq(cy_j/x_i)}{\Gampq(cty_j/x_i)},$$
$$ \prod_{i\in I,j\notin I}\frac{\thp(tx_j/x_i)}{\thp(x_j/x_i)}\prod_{1\leq i\neq j\leq n}\frac{\Gampq(tx_i/x_j)}{\Gampq(x_i/x_j)}
\mapsto\prod_{i\in I,j\notin I}\frac{\thp(tx_i/x_j)}{\thp(x_i/x_j)}\prod_{1\leq i\neq j\leq n}\frac{\Gampq(tx_i/x_j)}{\Gampq(x_i/x_j)}.
 $$
 This shows that
 \begin{align*}
(D^{(k)}Q_cf)(\mathbf y)&=\int_{\mathbb T^{n-1}_{y_1\dotsm y_n}}\sum_{\substack{I\subseteq \{1,\dots,n\},\\ |I|=k}}
\prod_{i\in I,j\notin I}\frac{\thp(tx_i/x_j)}{\thp(x_i/x_j)}
\prod_{i\in I}(T_{q,x_i}f)(\mathbf x)
 \\
&\quad\times\prod_{1\leq i,j\leq n}\frac{\Gampq(cy_j/x_i)}{\Gampq(cty_j/x_i)}
\prod_{1\leq i\neq j\leq n}\frac{\Gampq(tx_i/x_j)}{\Gampq(x_i/x_j)} \,|\dup \mathbf x|\\
&=(Q_cD^{(k)}f(\mathbf y).
\end{align*}

\end{document}